\documentclass[a4paper,fleqn,11pt]{article}

\usepackage{amsmath}
\usepackage{amsthm}
\usepackage{amssymb}
\usepackage{accents}

\usepackage[a4paper,top=3cm, bottom=3cm, left=3cm, right=3cm]{geometry}
\usepackage[shortlabels,inline]{enumitem}
\usepackage[pdftex]{graphicx}
\usepackage{subcaption}
\usepackage[pdftex,ocgcolorlinks,pagebackref=false]{hyperref}
\usepackage[affil-it]{authblk}

\setlist[enumerate,1]{label={(\roman*)}}

\usepackage[capitalise,nameinlink,noabbrev]{cleveref}

\Crefname{property}{Property}{Properties}

\theoremstyle{plain}
\newtheorem{theorem}{Theorem}[section]
\newtheorem{lemma}[theorem]{Lemma}
\newtheorem{proposition}[theorem]{Proposition}
\newtheorem{corollary}[theorem]{Corollary}
\newtheorem{remark}[theorem]{Remark}
\newtheorem{example}[theorem]{Example}

\theoremstyle{definition}
\newtheorem{definition}[theorem]{Definition}
\newtheorem{problem}[theorem]{Problem}

\newcommand{\naturals}{\mathbb{N}}
\newcommand{\positiveintegers}{\mathbb{N}_{>0}}
\newcommand{\nonnegativereals}{\mathbb{R}_{\ge 0}}

\newcommand{\distributions}[1][]{\mathcal{P}_{#1}}
\DeclareMathOperator{\probability}{Pr}

\newcommand{\entropy}{H}
\newcommand{\relativeentropy}[3][]{\mathop{D_{#1}}\mathopen{}\left(#2\middle\|#3\right)\mathclose{}}
\newcommand{\mutualinformation}{I}
\newcommand{\typeclass}[2]{T^{#1}_{#2}}

\newcommand{\graphs}{\mathcal{G}}

\newcommand{\typegraph}[3]{#1^{\strongproduct #2}[\typeclass{#2}{#3}]}
\renewcommand{\complement}[1]{\overline{#1}}
\newcommand{\disjointunion}{\sqcup}
\newcommand{\graphjoin}{+}
\newcommand{\strongproduct}{\boxtimes}
\newcommand{\costrongproduct}{*}
\newcommand{\adjacentorequal}{\simeq}

\DeclareMathOperator{\support}{supp}
\newcommand{\norm}[2][]{\left\|#2\right\|_{#1}}
\newcommand{\ball}[2]{B_{#1}(#2)}
\newcommand{\closedball}[2]{\overline{B}_{#1}(#2)}
\newcommand{\setbuild}[2]{\left\{#1\middle|#2\right\}}

\newcommand{\entropyrate}{\overline{\entropy}}
\newcommand{\Frate}{\overline{F}}
\newcommand{\Ornsteindistance}{\bar{d}}
\newcommand{\Wassersteindistance}{W_1}
\newcommand{\secondorderdistributions}[1][]{\mathcal{P}^{(2)}_{#1}}
\newcommand{\secondordertypeclass}[2]{T^{#1,(2)}_{#2}}

\title{Asymptotic equipartition property for a Markov source having ambiguous alphabet}
\author[1]{Tam\'as Tasn\'adi}
\author[2,3]{P\'eter Vrana}
\affil[1]{Department of Analysis, Institute of Mathematics, Budapest University of Technology and Economics, M\H uegyetem rkp. 3., H-1111 Budapest, Hungary.}
\affil[2]{Department of Geometry, Institute of Mathematics, Budapest University of Technology and Economics, M\H uegyetem rkp. 3., H-1111 Budapest, Hungary.}
\affil[3]{MTA-BME Lend\"ulet Quantum Information Theory Research Group, M\H uegyetem rkp. 3., H-1111 Budapest, Hungary}

\begin{document}
\maketitle

\begin{abstract}
We propose a generalization of the asymptotic equipartition property to discrete sources with an ambiguous alphabet, and prove that it holds for irreducible stationary Markov sources with an arbitrary distinguishability relation. Our definition is based on the limiting behavior of graph parameters appearing in a recent dual characterization of the Shannon capacity, evaluated at subgraphs of strong powers of the confusability graph induced on high-probability subsets. As a special case, our results give an information-theoretic interpretation of the graph entropy rate of such sources.
\end{abstract}

\section{Introduction}

A discrete stationary source over the alphabet $\mathcal{X}$ with entropy rate $\entropyrate$ satisfies the asymptotic equipartition property if for every large $n$ there is a subset of $\mathcal{X}^n$ of size $2^{n\entropyrate+o(n)}$, total probabilty arbitrarily close to $1$ and such that every sequence within the subset has probability $2^{-n\entropyrate+o(n)}$ (typical set). Forming the basis of many results in information theory, the asymptotic equipartition property for a sequence of independent and identically distributed random variables was proved by Shannon \cite{shannon1948mathematical}, subsequently extended to ergodic sources in the Shannon--McMillan--Breiman theorem \cite{mcmillan1953basic,breiman1957individual}. The most direct application of the asymptotic equipartition property is to source coding, as it implies that the simple coding scheme that assigns different codewords to the typical set and an arbitrary sequence to the remaining ones, requires approximately $\entropyrate$ bits per letter, and is asymptotically optimal.

In \cite{korner1973coding} K\"orner studied a coding problem for a discrete memoryless source that produces symbols which cannot be fully distinguished from each other. In this model, distinguishability is a symmetric relation on the alphabet, and two strings of equal length are distinguishable if they are distinguishable in at least one coordinate. The task is then to assign bit strings to the symbols emitted by the source in such a way that, with high probability, symbols that can be distinguished are mapped to different strings. If we encode the relation in a graph $G$ whose vertex set is the alphabet, and two symbols form an edge if they can be confused, then the confusability graph of $n$-letter strings will be $G^{\strongproduct n}$, the $n$-fold strong product of $G$ with itself. If the common distribution of the letters is $P$, then the number of bits required grows as $n\entropy(\overline{G},P)+o(n)$, where $\entropy(\complement{G},P)$ is the graph entropy of $\complement{G}$. The graph entropy of the complement can be expressed either as the entropy function of a certain convex corner \cite[3.1 Lemma]{csiszar1990entropy} (the vertex packing polytope), or as the limit of the (fractional) clique cover numbers of high-probability subsets of $G^{\strongproduct n}$ (in the complementary picture: chromatic numbers of such subsets of the conormal power of $G$) \cite{korner1973coding}:
\begin{equation}\label{eq:graphentropy}
\entropy(\complement{G},P)=\lim_{n\to\infty}\min_{\substack{S\subseteq V(G)^n  \\  P^{\otimes n}(S)\ge c}}\frac{1}{n}\log\overline{\chi}_f(G^{\strongproduct n}[S]),
\end{equation}
where $c\in(0,1)$ is arbitrary. For more on the properties and generalizations of the graph entropy, we refer to the surveys \cite{simonyi1995graph,simonyi2001perfect}.

In this coding problem, the (fractional) clique cover number serves as a way to measure the effective size of a set of partially distinguishable symbols, in the sense that at least $\log\overline{\chi}_f(G)$ bits are required to be able to map distinguishable symbols to distinct bit strings. More generally, given two alphabets with confusability graphs $G$ and $H$, we may consider encodings $\varphi:V(G)\to V(H)$ that map distinguishable pairs to distinguishable ones, i.e. require that $\varphi$ be a homomorphism $\complement{G}\to\complement{H}$. If such an encoding exists, then $\overline{\chi}_f(G)\le\overline{\chi}_f(H)$. Moreover, since $\overline{\chi}_f$ is multiplicative under the strong product (unlike $\overline{\chi}$), this inequality remains true even if we only assume the existence of a homomorphism $\complement{G^{\strongproduct n}}\to\complement{H^{\strongproduct n}}$, i.e. a block-encoding. However, the inequality between the fractional clique cover numbers is not sufficient for the existence of block-encodings in general. In a recent work Zuiddam introduced the asymptotic spectrum of graphs \cite{zuiddam2019asymptotic}, the set of nonnegative real graph parameters that are multiplicative under the strong product, additive under the disjoint union, and monotone in the sense that $f(G)\le f(H)$ if there is a homomorphism from $\complement{G}$ to $\complement{H}$, and showed that if the inequality $f(G)\le f(H)$ holds simultaneously for all such functions, then there exist homomorphisms $\complement{G^{\strongproduct n+o(n)}}\to\complement{H^{\strongproduct n}}$. Elements of the asymptotic spectrum of graphs, or spectral points, include the fractional clique cover number, the Lov\'asz number \cite{lovasz1979shannon}, the complement of the projective rank $\complement{\xi}_f(G)=\xi_f(\complement{G})$ \cite{cubitt2014bounds}, and the fractional Haemers bound $\mathcal{H}^\mathbb{F}_f(G)$ over any field $\mathbb{F}$ \cite{haemers1979some,blasiak2013graph,bukh2018fractional}. One may regard these as different ways to measure the effective size of the alphabet, which collectively characterize the possible encodings of large blocks.

From this point of view, \eqref{eq:graphentropy} is a generalization of the asymptotic equipartition property, and for an edgeless graph, it is indeed equivalent to it. Therefore we make the following definition (see \cref{def:generalAEP} below). A source with an ambiguous alphabet, characterized by a finite graph $G$ and a distibution $\mu$ on $V(G)^{\infty}$ satisfies the asymptotic equipartition property if for every element $f$ of the asymptotic spectrum of graphs, the limit
\begin{equation}
\lim_{n\to\infty}\min_{\substack{S\subseteq V(G)^n  \\  \mu_{12\ldots n}(S)\ge c}}\frac{1}{n}\log f(G^{\strongproduct n}[S])
\end{equation}
exists and is independent of $c$ for $c\in(0,1)$. In \cite{vrana2021probabilistic} it was shown that when $\mu$ is the joint distribution of a sequence of independent random variables with common distribution $P$, then this property holds. The value of the limit is $F(G,P)$, where $F$ is the (logarithmic) probabilistic refinement of $f$. For certain spectral points, such as the Lov\'asz number and the fractional clique cover number, the probabilistic refinement has a single-letter characterization.

In this work, we extend the asymptotic equipartition property in the above sense to ergodic Markov processes with an ambiguous alphabet. We characterize the limit in terms of the logarithmic probabilistic refinement evaluated at large blocks, in analogy with the entropy rate. In \cref{sec:ambiguous} we review the concept of a partially distinguishable alphabet and collect the relevant graph-theoretical notions, including the asymptotic spectrum of graphs and the properties of the probabilistic refinements of spectral points. In \cref{sec:AEP} we reformulate the asymptotic equipartition property for an arbitrary source in a form that motivates our generalization for ambiguous alphabets. In \cref{sec:Frate} we define and study a quantity analogous to the entropy rate, associated with the probabilistic refinement of a spectral point and a source with a partially distinguishable alphabet. Like the entropy rate, the new quantity exists for stationary sources, and is continuous with respect to the Ornstein distance. In \cref{sec:Markov} we prove our main result, that irreducible stationary Markov chains with an arbitrary distinguishability relation have the asymptotic equipartition property. In \cref{sec:hidden} we study processes that arise from observing a function of the output of a (hidden) process, and relate the asymptotic equipartition property for the observed source to that of the hidden process equipped with an induced confusability relation.

\section{Ambiguous alphabets}\label{sec:ambiguous}
We consider finite alphabets with a fixed symmetric relation with the interpretation that pairs of letters may or may not be distinguishable. We adopt the convention that the relation is given by the confusability graph: the vertices of this graph are the letters in the alphabet, and there is an edge between two distinct letters if they are confusable, and no edge if they are distinguishable. Note that the opposite convention is also used in the literature -- to translate the results, one needs to replace graphs with their complements, which interchanges the disjoint union ($\disjointunion$) with join ($\graphjoin$), the strong product ($\strongproduct$, also known as strong direct product, AND product, normal product) with the costrong product ($\costrongproduct$, disjunctive product, OR product, co-normal product), etc. We will use the notation $g\adjacentorequal h$ to mean that $g$ and $h$ are either adjacent or equal, e.g. the strong product is defined as $(g_1,g_2)\adjacentorequal(h_1,h_2)$ iff $g_1\adjacentorequal h_1$ and $g_2\adjacentorequal h_1$.

A confusability relation on an alphabet $\mathcal{X}$ induces a similar relation on the set $\mathcal{X}^n$ of $n$-letter strings: two strings are distinguishable iff there is at least one coordinate where they are distinguishable. In terms of the confusability graph $G=(\mathcal{X},E)$, this relation is given by the power $G^{\strongproduct n}$ with respect to the strong product. We note that a different, more general approach would be to start with a confusability relation on the set of infinite strings, and define finite strings to be confusable if they can be extended to confusable infinite strings (see \cite[Section 5.3]{boreland2020information}). Assuming that the relation is shift-invariant, this leads to a sequence of confusability graphs $G_n$ with vertex set $V(G_n)=\mathcal{X}^n$ subject to the compatiblity condition $G_m\strongproduct G_n\supseteq G_{n+m}$. We do not explore this setting here, but remark that some of the results below remain true when $G^{\strongproduct n}$ is replaced with $G_n$.

As a general rule, any coding problem where one is looking for an encoding, i.e. function $\varphi:\mathcal{X}\to\mathcal{Y}$ subject to certain constraints may be generalized by equipping the input and output alphabets with a confusability relation and requiring the functions to map distinguishable pairs into distinguishable ones, i.e. that $\varphi$ is a homomorphism between the complements of the graphs. In the case of block encodings we require that $\varphi:\complement{H^{\strongproduct n}}\to\complement{G^{\strongproduct Rn+o(n)}}$ is a homomorphism. We will write $H\le G$ if there exists a homomorphism $\complement{H}\to\complement{G}$.

In the special case when $H$ has no edge, the requirement is that distinct inputs are encoded as distinguishable outputs, i.e. we are essentially looking for large independent sets in $G^{\strongproduct n}$. This is the problem of zero-error communication through a discrete memoryless channel \cite{shannon1956zero}. In the opposite situation, where $H$ is arbitrary and $G$ is edgeless, a homomorphism $\varphi:\complement{H^{\strongproduct n}}\to\complement{G^{\strongproduct Rn+o(n)}}$ is essentially a covering of $H^{\strongproduct n}$ by $\lvert G\rvert^{Rn+o(n)}$ cliques, which is possible precisely when the fractional clique cover number of $H$ is at most $\lvert G\rvert^R$. In \cite{korner1973coding} K\"orner considered a stochastic version of this problem, where the input follows a (memoryless) distribution $P^{\otimes n}$ for some $P\in\distributions(V(H))$, and one only needs to encode typical strings. In this case the rate is characterized by the graph entropy $\entropy(\complement{H},P)$.

The asymptotic spectrum of graphs $\Delta(\graphs)$ is the set of $\nonnegativereals$-valued graph parameters $f$ satisfying the following properties \cite{zuiddam2019asymptotic}:
\begin{enumerate}
\item $f(G\disjointunion H)=f(G)+f(H)$;
\item $f(G\strongproduct H)=f(G)f(H)$;
\item if $H\le G$, then $f(H)\le f(G)$;
\item $f(\complement{K_d})=d$.
\end{enumerate}
Elements of $\Delta(\graphs)$ are called spectral points. These properties imply that whenever a block encoding $\varphi:\complement{H^{\strongproduct n}}\to\complement{G^{\strongproduct Rn+o(n)}}$ exists for infinitely many block lengths $n$, then $f(H)\le f(G)^R$ holds for every spectral point $f\in\Delta(\graphs)$. Remarkably, the converse is also true in the sense that if $f(H)\le f(G)$ holds for all spectral points $f$, then there exist block encodings $\complement{H^{\strongproduct n}}\to\complement{G^{\strongproduct Rn+o(n)}}$ for all $n$ \cite[Theorem 1.1.]{zuiddam2019asymptotic}.

As mentioned in the introduction, a useful way to think about the spectral points is that they generalize, in various ways, the cardinality of the alphabet. With a stochastic model of a source in mind, one may consider an ambiguous alphabet together with a probability distribution on the letters, which gives rise to an i.i.d. source. In analogy with the coding problem leading to the graph entropy, one can then ask how the effective size of the typical sets grow as the block length increases. This is captured by the probabilistic refinement of the spectral points \cite{vrana2021probabilistic}. A possible definition of these is the following. Let $G$ be a graph, $P\in\distributions(V(G))$ and $P^{(n)}$ a sequence of $n$-types converging to $P$. Then it can be shown that
\begin{equation}
F(G,P):=\lim_{n\to\infty}\frac{1}{n}\log f(\typegraph{G}{n}{P^{(n)}})
\end{equation}
is well-defined, i.e. the limit exists and depends only on $G$ and the limiting distribution $P$. Moreover, the spectral point $f$ may be recovered as $\log f(G)=\max_{P\in\distributions(V(G))}F(G,P)$ and the probabilistic refinements of spectral points are characterized by the following properties \cite[Theorems 1.1 and 1.2]{vrana2021probabilistic}:
\begin{enumerate}
\item $P\mapsto F(G,P)$ is concave;
\item if $P\in\distributions(V(G)\times V(H))$ with marginals $P_G$ and $P_H$, then
\begin{equation}\label{eq:subadditivity}
F(G\strongproduct H,P)\le F(G,P_G)+F(H,P_H)\le F(G\strongproduct H,P)+\mutualinformation(G:H)_P,
\end{equation}
where $\mutualinformation(G:H)_P=\entropy(P_G)+\entropy(P_H)-\entropy(P)$;
\item if $P_G\in\distributions(V(G))$, $P_H\in\distributions(V(H))$ and $p\in[0,1]$, then
\begin{equation}
F(G\disjointunion H,pP_G\oplus(1-p)P_H)=pF(G,P_G)+(1-p)F(H,P_H)+h(p),
\end{equation}
where $h(p)=-p\log p-(1-p)\log(1-p)$;
\item\label[property]{it:Fmonotone} if $\varphi:\complement{H}\to\complement{G}$ is a homomorphism and $P\in\distributions(V(H))$, then $F(H,P)\le F(G,\varphi_*(P))$.
\end{enumerate}

\begin{example}\leavevmode
\begin{enumerate}
\item The fractional clique cover number is an element of $\Delta(\graphs)$. Its probabilistic refinement is the graph entropy of the complement.
\item The probabilistic refinement of the Lov\'asz number is equal to the functional studied by Marton in \cite{marton1993shannon}.
\end{enumerate}
\end{example}

It is not a coincidence that these properties of the probabilistic refinements (and some others, such as the continuity estimate \cite[Proposition 3.3 (iv)]{vrana2021probabilistic}) resemble those of the Shannon entropy, and indeed these functionals arise as special cases of an entropy concept introduced in \cite{csiszar1990entropy}, the entropy function of a convex corner. One of the aims of the present paper is to strengthen the entropic interpretation of the functionals arising as probabilistic refinements of spectral points, and in particular to explore their properties in connection with sequences of dependent random variables over ambiguous alphabets. As an example, let us see how arguably the simplest such property of the Shannon entropy generalizes: that the joint entropy of a pair of random variables is not less than any of the individual entropies or, in other words, that the conditional entropy is nonnegative.
\begin{proposition}\label{prop:conditionalFnonnegative}
Let $G,H$ be graphs and $P\in\distributions(V(G)\times V(H))$. Then $F(H,P_H)\le F(G\strongproduct H,P)$.
\end{proposition}
\begin{proof}
If $\psi:V(H)\to V(G)$ is any function, then $\varphi^{(\psi)}(h):=(\psi(h),h)$ is a homomorphism $\complement{H}\to\complement{G\strongproduct H}$, therefore $F(H,P_H)\le F(G\strongproduct H,\varphi^{(\psi)}_*(P_H))$. $P$ is in the convex hull of $\setbuild{\varphi^{(\psi)}_*(P_H)}{\psi:V(H)\to V(G)}$ and $F$ is convace, therefore $F(H,P_H)\le F(G\strongproduct H,P)$.
\end{proof}
\begin{remark}
An alternative, ``low-level'' proof proceeds by first considering joint $n$-types $P\in\distributions[n](V(G)\times V(H))$ and noting that $\typegraph{H}{n}{P_H}\le\typegraph{(G\strongproduct H)}{n}{P}$, which can be seen by mapping each $y\in\typeclass{n}{P_H}$ to $(x,y)$ for any $x$ such that the joint type of $(x,y)$ is $P$.
\end{remark}

One may reasonably conjecture that the strong subadditivity (or submodularity) property also holds, but a proof has so far eluded us.
\begin{problem}\label{prob:strongsubadditivity}
Is it true that for every triple of graphs $G,H,K$ and $P\in\distributions(V(G)\times V(H)\times V(K))$ the inequality
\begin{equation}
F(G\strongproduct H,P_{GH})+F(H\strongproduct K,P_{HK})\ge F(H,P_{H})+F(G\strongproduct H\strongproduct K,P)
\end{equation}
holds?
\end{problem}

\section{Asymptotic equipartition property}\label{sec:AEP}

In the following, $\mathcal{X}$ is a finite set (alphabet) equipped with the $\sigma$-algebra $2^\mathcal{X}$, and $\mathcal{X}^\infty$ denotes the set of infinite words $x_1x_2\ldots$ over the alphabet $\mathcal{X}$ with the product $\sigma$-algebra. Let $\mu$ be a probability measure on $\mathcal{X}^\infty$, and let $\mu_1,\mu_2,\mu_{12},\mu_{12\ldots n},\ldots$ denote its various marginals. We say that $\mu$ has the asymptotic equipartition property if, considered as the joint distribution of the sequence of random variables $X_1,X_2,\ldots$, the random variables
\begin{equation}
-\frac{1}{n}\log\mu_{12\ldots n}(X_1,X_2,\ldots,X_n)
\end{equation}
converge to a constant $h$ (entropy rate) in probability. In detail, for every $\epsilon>0$,
\begin{equation}
\lim_{n\to\infty}\probability\left(\left\lvert-\frac{1}{n}\log\mu_{12\ldots n}(X_1,X_2,\ldots,X_n)-h\right\rvert>\epsilon\right)=0.
\end{equation}
In \cite{shannon1948mathematical} Shannon proved the asymptotic equipartition property for i.i.d. sources, while McMillan \cite{mcmillan1953basic} and Breiman \cite{breiman1957individual} showed that ergodic sources also have the asymptotic equipartition property.

An important consequence of the asymptotic equipartition property is that, for each $n$, the strings in $\mathcal{X}^n$ can be divided into a typical subset $A^{(n)}_\epsilon$ and its complement in such a way that $\lvert A^{(n)}_\epsilon\rvert=2^{nh+o(n)}$, $\mu_{1\ldots n}(A^{(n)}_\epsilon)=1-o(1)$ and every string in $A^{(n)}_\epsilon$ has probability $2^{-nh+o(n)}$. In the following we show that it is possible to reformulate the asymptotic equipartition property solely in terms of the cardinalities of such high-probability subsets (the implication in one direction for i.i.d. sequences is \cite[Theorem 3.3.1]{cover2012elements}).
\begin{proposition}\label{prop:AEPreformulation}
The following statements are equivalent:
\begin{enumerate}
\item\label{it:AEP} $\mu$ has the asymptotic equipartition property with entropy rate $h$;
\item\label{it:highprobabilitysubsetsize} $\displaystyle\lim_{n\to\infty}\min_{\substack{S\subseteq\mathcal{X}^n  \\  \mu_{12\ldots n}(S)\ge c}}\frac{1}{n}\log\,\lvert S\rvert=h$ for all $c\in(0,1)$.
\end{enumerate}
\end{proposition}
\begin{proof}
\ref{it:AEP}$\implies$\ref{it:highprobabilitysubsetsize}: Suppose that $\mu$ has the asymptotic equipartition property and let $c\in(0,1)$ and $\epsilon>0$. Consider the set
\begin{equation}
A^{(n)}_\epsilon=\setbuild{x\in\mathcal{X}^n}{\mu_{1\ldots n}(x)\in[2^{-n(h+\epsilon)},2^{-n(h-\epsilon)}]}.
\end{equation}
Then $\mu(A^{(n)}_\epsilon)\to 1$, therefore for every sufficiently large $n$ this probability exceeds $c$. The total probability of $A^{(n)}_\epsilon$ is less than $1$, but also at least $\lvert A^{(n)}_\epsilon\rvert 2^{-n(h+\epsilon)}$. It follows that
\begin{equation}
\begin{split}
\limsup_{n\to\infty}\min_{\substack{S\subseteq\mathcal{X}^n  \\  \mu_{12\ldots n}(S)\ge c}}\frac{1}{n}\log\,\lvert S\rvert
 & \le \limsup_{n\to\infty}\frac{1}{n}\log\,\lvert A^{(n)}_\epsilon\rvert  \\
 & \le \limsup_{n\to\infty}\frac{1}{n}\log 2^{n(h+\epsilon)}  \\
 & = h+\epsilon.
\end{split}
\end{equation}

For a lower bound, split an arbitrary subset $S\subseteq\mathcal{X}^n$ into its intersection with $A^{(n)}_\epsilon$ and its complement and bound the total probability of both parts as
\begin{equation}
\begin{split}
\mu_{12\ldots n}(S)
 & = \mu_{12\ldots n}(S\setminus A^{(n)}_\epsilon)+\mu_{12\ldots n}(S\cap A^{(n)}_\epsilon)  \\
 & \le \mu_{12\ldots n}(\mathcal{X}^n\setminus A^{(n)}_\epsilon)+\lvert S\vert 2^{-n(h-\epsilon)}.
\end{split}
\end{equation}
Note that the first term on the right hand side converges to $0$. If $\mu_{12\ldots n}(S)\ge c$, this inequality implies
\begin{equation}
\left(c-\mu_{12\ldots n}(\mathcal{X}^n\setminus A^{(n)}_\epsilon)\right)2^{n(h-\epsilon)} \le \lvert S\vert,
\end{equation}
therefore
\begin{equation}
\begin{split}
\liminf_{n\to\infty}\min_{\substack{S\subseteq\mathcal{X}^n  \\  \mu_{12\ldots n}(S)\ge c}}\frac{1}{n}\log\,\lvert S\rvert
 & \ge \liminf_{n\to\infty}\frac{1}{n}\log\left(c-\mu_{12\ldots n}(\mathcal{X}^n\setminus A^{(n)}_\epsilon)\right)2^{n(h-\epsilon)}  \\
 & = h-\epsilon.
\end{split}
\end{equation}
Since both the lower and the upper bound holds for every $\epsilon>0$, the limit exists and is equal to $h$, independently of $c$.

\ref{it:highprobabilitysubsetsize}$\implies$\ref{it:AEP}:
In the other direction, suppose that for every $c\in(0,1)$ the equality
\begin{equation}
\lim_{n\to\infty}\min_{\substack{S\subseteq\mathcal{X}^n  \\  \mu_{12\ldots n}(S)\ge c}}\frac{1}{n}\log\,\lvert S\rvert=h
\end{equation}
holds. Let $\delta>0$ and for all sufficiently large $n$, choose minimal-size subsets $S^{(n)},T^{(n)},U^{(n)}\subseteq\mathcal{X}^n$ subject to $\mu_{12\ldots n}(S^{(n)})\ge \delta$, $\mu_{12\ldots n}(T^{(n)})\ge 1-2\delta$, and $\mu_{12\ldots n}(U^{(n)})\ge 1-\delta$. We may assume that each subset contains the highest-probability strings for a given cardinality and that, in addition, $S^{(n)}\subseteq T^{(n)}\subseteq U^{(n)}$ (which is an extra assumption if there are equal probabilities). Let $s_0^{(n)}\in S^{(n)}$ and $t_0^{(n)}\in T^{(n)}$ be one of the elements with the lowest probability, and let $B^{(n)}=T^{(n)}\setminus (S^{(n)}\setminus\{s_0\})$ and $C^{(n)}=U^{(n)}\setminus (T^{(n)}\setminus\{t_0\})$. Note that, by construction, $s_0^{(n)}$ and $t_0^{(n)}$ are one of the elements with the highest probability in $B^{(n)}$ and $C^{(n)}$, respectively.

By the minimality of $\lvert S^{(n)}\rvert$, $\mu_{12\ldots n}(S^{(n)}\setminus\{s_0^{(n)}\})<\delta$ and therefore
\begin{equation}
\begin{split}
\mu_{12\ldots n}(B^{(n)})
 & = \mu_{12\ldots n}(T^{(n)})-\mu_{12\ldots n}(S^{(n)}\setminus\{s_0^{(n)}\})  \\
 & \ge 1-3\delta.
\end{split}
\end{equation}
By a similar reasoning, we have $\mu_{12\ldots n}(C^{(n)})\ge\delta$.

The estimates
\begin{equation}
\begin{split}
1
 & \ge \mu_{12\ldots n}(S^{(n)})  \\
 & \ge \lvert S^{(n)}\rvert\min_{s\in S^{(n)}}\mu_{12\ldots n}(\{s\})  \\
 & = \lvert S^{(n)}\rvert\mu_{12\ldots n}(\{s_0^{(n)}\})  \\
 & = \lvert S^{(n)}\rvert\max_{s\in B^{(n)}}\mu_{12\ldots n}(\{s\}).
\end{split}
\end{equation}
and
\begin{equation}
\begin{split}
\delta
 & \le \mu_{12\ldots n}(C^{(n)})  \\
 & \le \lvert C^{(n)}\rvert\max_{s\in C^{(n)}}\mu_{12\ldots n}(\{s\})  \\
 & \le \lvert U^{(n)}\rvert\mu_{12\ldots n}(\{t_0^{(n)}\})  \\
 & = \lvert U^{(n)}\rvert\min_{s\in B^{(n)}}\mu_{12\ldots n}(\{s\})
\end{split}
\end{equation}
imply that for all $s\in B^{(n)}$ we have $\delta\lvert U^{(n)}\rvert^{-1}\le\mu_{12\ldots n}(\{s\})\le\lvert S^{(n)}\rvert^{-1}$, i.e.
\begin{equation}
-\frac{1}{n}\log\delta+\frac{1}{n}\log\,\lvert U^{(n)}\rvert\ge-\frac{1}{n}\log\mu_{12\ldots n}(\{s\})\ge\frac{1}{n}\log\,\lvert S^{(n)}\rvert.
\end{equation}

By assumption, both the lower and the upper bound has limit $h$. Therefore for every $\epsilon>0$, the probability that $-\frac{1}{n}\log\mu_{12\ldots n}(X_1,X_2,\ldots,X_n)$ is in $[h-\epsilon,h+\epsilon]$ is at least $\mu_{12\ldots n}(B^{(n)})\ge 1-3\delta$, for every sufficiently large $n$. As $\delta>0$ was arbitrary, the asymptotic equipartition property holds.
\end{proof}

\Cref{prop:AEPreformulation} offers a way to generalize the asymptotic equipartition property to sources with an ambiguous alphabet by replacing the minimal-cardinality subset with the induced subgraph attaining the minimal value for each spectral point, subject to the probability constraint:
\begin{definition}\label{def:generalAEP}
Let $G$ be a graph (encoding the confusability relation of an ambiguous alphabet $V(G)$) and $\mu$ a probability measure on $V(G)^\infty$. The pair $(G,\mu)$ has the asymptotic equipartition property if for every $f\in\Delta(\graphs)$ the limit
\begin{equation}
\lim_{n\to\infty}\min_{\substack{S\subseteq V(G)^n  \\  \mu_{12\ldots n}(S)\ge c}}\frac{1}{n}\log f(G^{\strongproduct n}[S])
\end{equation}
exists and is independent of $c\in(0,1)$.
\end{definition}

\begin{example}\leavevmode
\begin{enumerate}
\item If $G=\complement{K_d}$ (every pair of distinct symbols is distinguishable) then $f(G^{\strongproduct n}[S])=\lvert S\rvert$ holds for every spectral point $f$, therefore \cref{def:generalAEP} recovers the asymptotic equipartition property of $\mu$ in the usual sense.
\item If $G=K_d$ (all letters are confusable) then $f(G^{\strongproduct n}[S])=1$ for all nonempty $S\subseteq V(G)^n$, therefore the pair $(G,\mu)$ has the asymptotic equipartition property for \emph{every} distribution $\mu$, and the common limit is $0$ for every spectral point.
\item If $G$ is an arbitrary nonempty graph and $P\in\distributions(V(G))$, then $(G,P^\infty)$ has the asymptotic equipartition property by \cite[Proposition 3.4]{vrana2021probabilistic}, and the common limit is equal to $F(G,P)$. The special case when $f$ is the fractional clique cover number was proved in \cite{korner1973coding}, identifying the common limit as the graph entropy (in the complementary picture).
\end{enumerate}
\end{example}

For theoretical purposes it is sometimes convenient to rescale the time variable in the sense that we regard a block of output symbols as a single symbol in a product alphabet (e.g. when the dependence between the blocks is simpler than that between the individual random variables). We show that this transformation does not affect the existence of the limit in \cref{def:generalAEP}.
\begin{proposition}\label{prop:blockAEP}
Let $G$ be a graph and $\mu$ a probability measure on $V(G)^\infty$. Then for every $c\in(0,1)$ and $k\in\positiveintegers$ we have
\begin{align}
\liminf_{n\to\infty}\min_{\substack{S\subseteq V(G)^{kn}  \\  \mu_{12\ldots kn}(S)\ge c}}\frac{1}{kn}\log f(G^{\strongproduct kn}[S]) & =\liminf_{n\to\infty}\min_{\substack{S\subseteq V(G)^n  \\  \mu_{12\ldots n}(S)\ge c}}\frac{1}{n}\log f(G^{\strongproduct n}[S])
\intertext{and}
\limsup_{n\to\infty}\min_{\substack{S\subseteq V(G)^{kn}  \\  \mu_{12\ldots kn}(S)\ge c}}\frac{1}{kn}\log f(G^{\strongproduct kn}[S]) & =\limsup_{n\to\infty}\min_{\substack{S\subseteq V(G)^n  \\  \mu_{12\ldots n}(S)\ge c}}\frac{1}{n}\log f(G^{\strongproduct n}[S]).
\end{align}
\end{proposition}
\begin{proof}
Consider the sequence $a_n=\min_{\substack{S\subseteq V(G)^n  \\  \mu_{12\ldots n}(S)\ge c}}\log f(G^{\strongproduct n}[S])$ and let $0\le m\le n$. Let $S\subseteq V(G)^m$ such that $\mu_{12\ldots m}(S)\ge c$ and $a_m=\log f(G^{\strongproduct m}[S])$. Then $\mu_{12\ldots n}(S\times V(G)^{n-m})=\mu_{12\ldots m}(S)\ge c$, therefore
\begin{equation}
\begin{split}
a_n
 & \le \log f(G^{\strongproduct n}[S\times V(G)^{n-m}])  \\
 & = \log f(G^{\strongproduct m}[S])f(G^{\strongproduct n-m})  \\
 & = a_m+(n-m)\log f(G).
\end{split}
\end{equation}
We use this inequality to estimate $a_n/n$ in terms of the similar fractions with indices that are the closest multiples of $k$:
\begin{equation}
\frac{k\lceil\frac{n}{k}\rceil}{n}\frac{a_{k\lceil\frac{n}{k}\rceil}}{k\lceil\frac{n}{k}\rceil}-\frac{k\log f(G)}{n}\le\frac{a_n}{n}\le\frac{k\lfloor\frac{n}{k}\rfloor}{n}\frac{a_{k\lfloor\frac{n}{k}\rfloor}}{k\lfloor\frac{n}{k}\rfloor}+\frac{k\log f(G)}{n}.
\end{equation}
The claim follows by taking the limit inferior and limit superior of these inequalities.
\end{proof}

\section{$F$-rate}\label{sec:Frate}

If $\mu$ is ergodic, then the Shannon--McMillan--Breimann theorem states that
\begin{equation}
\lim_{n\to\infty}\min_{\substack{S\subseteq\mathcal{X}^n  \\  \mu_{12\ldots n}(S)\ge c}}\frac{1}{n}\log\,\lvert S\rvert=\lim_{n\to\infty}\frac{1}{n}\entropy(\mu_{12\ldots n})=:\entropyrate(\mu)
\end{equation}
for all $c\in(0,1)$, where the right hand side is the entropy rate of the stochastic process $X_1,X_2,\ldots$ with joint distribution $\mu$. In contrast, \cref{def:generalAEP} does not involve any other characterization of the limit. An obvious candidate for its value is the limit of $F(G^{\strongproduct n},\mu_{12\ldots n})/n$, assuming it exists, but in this generality we can only establish a lower bound as follows.
\begin{lemma}\label{lem:inducedsubgraphfbound}
Let $G$ be a graph, $P\in\distributions(V(G))$ and $S\subseteq V(G)$. Then
\begin{equation}
\log f(G[S])\ge F(G,P)-(1-P(S))\log\,\lvert V(G)\rvert-1.
\end{equation}
\end{lemma}
\begin{proof}
Decompose $P$ as $P=P(S)Q_1+(1-P(S))Q_2$ where $Q_1,Q_2$ are probability distributions such that $\support Q_1\subseteq S$ and $\support Q_2\cap S=\emptyset$. Then
\begin{equation}
\begin{split}
F(G,P)
 & \le F(G[S]\disjointunion G[V(G)\setminus S],P(S)Q_1+(1-P(S))Q_2)  \\
 & = P(S)F(G,Q_1)+(1-P(S))F(G,Q_2)+h(P(S))  \\
 & \le \log f(G[S])+(1-P(S))\log\,\lvert V(G)\rvert+1,
\end{split}
\end{equation}
which is equivalent to the stated inequality.
\end{proof}

\begin{corollary}
Let $G$ be a graph and $\mu$ a probability measure on $V(G)^\infty$. If $(G,\mu)$ has the asymptotic equipartition property, then for every $f\in\Delta(\graphs)$ the inequality
\begin{equation}
\lim_{n\to\infty}\min_{\substack{S\subseteq V(G)^n  \\  \mu_{12\ldots n}(S)\ge c}}\frac{1}{n}\log f(G^{\strongproduct n}[S])\ge\limsup_{n\to\infty}\frac{1}{n}F(G^{\strongproduct n},\mu_{12\ldots n})
\end{equation}
holds (where $c\in(0,1)$ is arbitrary).
\end{corollary}
\begin{proof}
Applying \cref{lem:inducedsubgraphfbound} to the graph $G^{\strongproduct n}$ and every subset $S\subseteq V(G)^n$ satisfying $\mu_{12\ldots n}(S)\ge c$ gives
\begin{equation}
\min_{\substack{S\subseteq V(G)^n  \\  \mu_{12\ldots n}(S)\ge c}}\frac{1}{n}\log f(G^{\strongproduct n}[S])\ge\frac{1}{n}F(G^{\strongproduct n},\mu_{12\ldots n})-(1-c)\log\,\lvert V(G)\rvert-\frac{1}{n}.
\end{equation}
We obtain the statement by taking the limit superior of both sides and using that the left hand side is a limit and independent of $c\in(0,1)$.
\end{proof}

This bound motivates the following definition.
\begin{definition}
Let $G$ be a graph and $\mu$ a distribution on $V(G)^\infty$. Given a spectral point $f$ with probabilistic refinement $F$, the $F$-rate of $(G,\mu)$ is the limit
\begin{equation}
\Frate(G,\mu)=\lim_{n\to\infty}\frac{1}{n}F(G^{\strongproduct n},\mu_{12\ldots n}),
\end{equation}
assuming it exists.
\end{definition}

\begin{problem}
Suppose that $(G,\mu)$ has the asymptotic equipartition property. Does this imply that $\Frate(G,\mu)$ exists and is equal to the common limit in \cref{def:generalAEP}?
\end{problem}

\subsection{Stationary sources}
Next we restrict attention to stationary sources, i.e. assume that $\mu_{m+1,m+2,\ldots,m+n}=\mu_{12\ldots n}$ for all $m,n\in\naturals$. In complete analogy with the entropy rate, the value of $F$ on initial segments forms a subadditive sequence:
\begin{equation}
\begin{split}
F(G^{\strongproduct(m+n)},\mu_{12\ldots(m+n)})
 & = F(G^{\strongproduct m}\strongproduct G^{\strongproduct n},\mu_{12\ldots(m+n)})  \\
 & \le F(G^{\strongproduct m},\mu_{12\ldots m})+F(G^{\strongproduct n},\mu_{m+1,m+2,\ldots,m+n})  \\
 & = F(G^{\strongproduct m},\mu_{12\ldots m})+F(G^{\strongproduct n},\mu_{12\ldots n}),
\end{split}
\end{equation}
where the inequality uses \eqref{eq:subadditivity} and the second equality uses that $\mu$ is stationary. By the Fekete lemma, this proves the following:
\begin{proposition}\label{prop:Frateinfimum}
If $G$ is arbitrary and $\mu$ is stationary on $V(G)^n$, then the $F$-rate exists and is equal to the infimum:
\begin{equation}\label{eq:Frateinf}
\Frate(G,\mu)=\lim_{n\to\infty}\frac{1}{n}F(G^{\strongproduct n},\mu_{12\ldots n})=\inf_{n\ge 1}\frac{1}{n}F(G^{\strongproduct n},\mu_{12\ldots n}).
\end{equation}
\end{proposition}
\begin{remark}
If \cref{prob:strongsubadditivity} has a positive answer, then for stationary $\mu$ the limit
\begin{equation}
\lim_{n\to\infty}\left(F(G^{\strongproduct n},\mu_{12\ldots n})-F(G^{\strongproduct(n-1)},\mu_{12\ldots n-1})\right)
\end{equation}
exists and is equal to $\Frate(G,\mu)$.
\end{remark}

As a consequence of \cref{prop:Frateinfimum}, we can bound the $F$-rate from above by determining (or estimating) $F(G^{\strongproduct n},\mu_{12\ldots n})$ for some $n$. For computational purposes it is desirable to have a sequence of lower bounds converging to the $F$-rate as well. Such a sequence can indeed be given provided that the entropy rate of $\mu$ is known. To this end, consider the sequence
\begin{equation}
b_n=\entropy(\mu_{12\ldots n})-F(G^{\strongproduct n},\mu_{12\ldots n})\ge 0.
\end{equation}
It satisfies
\begin{equation}
\begin{split}
b_{n+m}
 & = \entropy(\mu_{12\ldots(m+n)})-F(G^{\strongproduct(m+n)},\mu_{12\ldots(m+n)})  \\
 & \le \entropy(\mu_{12\ldots(m+n)})-F(G^{\strongproduct m},\mu_{12\ldots m})-F(G^{\strongproduct n},\mu_{m+1,m+2,\ldots,m+n)})  \\
 &\qquad+I(1\ldots m:m+1,\ldots,m+n)_{\mu}  \\
 & = \entropy(\mu_{12\ldots m})-F(G^{\strongproduct m},\mu_{12\ldots m})+\entropy(\mu_{m+1,m+2\ldots m+n})-F(G^{\strongproduct n},\mu_{m+1,m+2\ldots m+n})  \\
 & = b_n+b_m
\end{split}
\end{equation}
by \eqref{eq:subadditivity} and using that $\mu$ is stationary. Therefore
\begin{equation}\label{eq:Fratesup}
\begin{split}
\Frate(G,\mu)
 & = \entropyrate(\mu)-\lim_{n\to\infty}\frac{b_n}{n}  \\
 & = \entropyrate(\mu)-\inf_{n\ge 1}\frac{b_n}{n}  \\
 & = \sup_{n\ge 1}\left[\entropyrate(\mu)-\frac{1}{n}\left(\entropy(\mu_{12\ldots n})-F(G^{\strongproduct n},\mu_{12\ldots n})\right)\right].
\end{split}
\end{equation}

\Cref{eq:Frateinf,eq:Fratesup} imply that the $F$-rate is computable whenever $F$, $\mu$ and the entropy rate are computable. However, as the following example indicates, one should not expect to find a simple formula for the limit (say, in terms of finitely many evaluations of $F$), even for seemingly simple graphs and processes.
\begin{example}[hidden Markov model]\label{ex:HMM}
Let $\mathcal{Z}$ be a finite set and $\mu$ the distribution of a stationary Markov process $Z_1,Z_2,\ldots$ with state space $\mathcal{Z}$. Let $\varphi:\mathcal{Z}\to\mathcal{X}$ be a function. Then $\varphi(Z_1),\varphi(Z_2),\ldots$ is a hidden Markov model with distribution $\varphi^{\infty}_*(\mu)$, with hidden states $Z_1,Z_2,\ldots$ and observed output $X_1,X_2,\ldots$ (the formally more general concept, when $X_i$ is allowed to be a random function of $Z_i$, conditionally independent of all the other variables, can also be cast in this form at the cost of enlarging the alphabet for the hidden process).

Let $G$ be the graph with vertex set $\mathcal{Z}$ and such that distinct vertices $z_1$ and $z_2$ form an edge iff $\varphi(z_1)=\varphi(z_2)$. Then $G^{\strongproduct n}$ is a disjoint union of cliques, therefore the probabilistic refinement of \emph{any} spectral point $f$ is equal to the entropy of the induced distribution on the set of components, i.e. $F(G,P)=\entropy(\varphi^n_*(P))$. This means that the $F$-rate of $(G,\mu)$ is equal to the entropy rate of $\varphi^{\infty}_*(\mu)$, for which no closed-form expression is known \cite{blackwell1959entropy}.
\end{example}

\subsection{Continuity}

The entropy rate of a stationary process is continuous with respect to the Ornstein distance \cite{gray1977maximum,gray2011entropy,muraki1992note,polyanskiy2016wasserstein}. In this section we prove a similar continuity estimate for the $F$-rate of stationary processes, for a fixed graph $G$ and spectral point $f$. First we briefly recall the definition of the Ornstein distance \cite{ornstein1973application,ornstein1974ergodic}.

Given a metric $d$ on a finite set $\mathcal{X}$, the Wasserstein distance induced by $d$ on $\distributions(\mathcal{X})$ is defined as
\begin{equation}
\Wassersteindistance(P,Q)=\min_{\substack{R\in\distributions(\mathcal{X}\times\mathcal{X})  \\  R_1=P  \\  R_2=Q}}\sum_{x,x'\in\mathcal{X}}R(x,x')d(x,x'),
\end{equation}
i.e. the minimum of the expected distance of the two components over all possible couplings of $P$ and $Q$.

Let now the set be of the form $\mathcal{X}^n$ for some finite alphabet $\mathcal{X}$ and $n\in\positiveintegers$, and consider the normalized Hamming distance
\begin{equation}
d_H(x,x')=\frac{1}{n}\left\lvert\setbuild{i\in[n]}{x_i\neq x'_i}\right\rvert.
\end{equation}
The Wasserstein distance induced by $d_H$ is called the Ornstein distance or d-bar distance, denoted by $\Ornsteindistance$. In addition, if $\mu$ and $\nu$ are stationary probability distributions on $\mathcal{X}^\infty$, then we define their distance to be
\begin{equation}
\Ornsteindistance(\mu,\nu)=\lim_{n\to\infty}\Ornsteindistance(\mu_{12\ldots n},\nu_{12\ldots n}).
\end{equation}
We note that the expected value of $d_H(X,X')$ when $X,X'$ are $\mathcal{X}^n$-valued random variables may also be written as
\begin{equation}
\frac{1}{n}\sum_{i=1}^n\probability(X_i\neq X'_i),
\end{equation}
and in the stationary case the distance is equal to the minimum probability that the first symbols are different over all stationary couplings of the two measures.

We start with a generalization of the continuity estimate \cite[Proposition 3.3 (iv)]{vrana2021probabilistic} to what can be viewed as a conditional version of $F$ (although with no known direct information-theoretic interpretation). It will be convenient to introduce the function $\eta(t)=(1+t)h(1/(1+t))$, which is continuous, concave and increasing on $[0,1]$, and satisfies $\eta(0)=0$ and $\eta(1)=2$.
\begin{lemma}
Let $G,H$ be graphs and $P,Q\in\distributions(V(G)\times V(H))$. Then
\begin{equation}\label{eq:conditionalFcontinuity}
\left\lvert\left(F(G\strongproduct H,P)-F(H,P_H)\right)-\left(F(G\strongproduct H,Q)-F(H,Q_H)\right)\right\rvert\le\delta\log f(G)+2\eta(\delta),
\end{equation}
where $\delta=\frac{1}{2}\norm[1]{P-Q}$ and e.g. $P_G$, $P_H$ denote the marginals of $P$ on $V(G)$ and $V(H)$.
\end{lemma}
\begin{proof}
The idea of the proof is similar to those in \cite{alicki2004continuity,winter2016tight,vrana2021probabilistic}. Consider the functions
\begin{align}
A(P) & = F(G\strongproduct H,P)+\entropy(P_H)-F(H,P_H)  \\
B(P) & = \entropy(P)-F(G\strongproduct H,P)+F(H,P_H).
\end{align}
Both functions are concave by \cite[Proposition 3.3]{vrana2021probabilistic}. Note also that $0\le F(G\strongproduct H,P)-F(H,P_H)\le F(G,P_G)\le\log f(G)$ by \eqref{eq:subadditivity} and \cref{prop:conditionalFnonnegative}.

Let $P'=\delta^{-1}(P-Q)_-$ and $Q'=\delta^{-1}(P-Q)_+$ be the normalized negative and positive parts so that $P',Q'\in\distributions(V(G)\times V(H))$ and $\lambda P'+(1-\lambda)P=\lambda Q'+(1-\lambda)Q=:R$ with $\lambda=\delta/(1+\delta)$. By concavity,
\begin{align}
A(R) & \ge \lambda A(P')+(1-\lambda)A(P)  \\
B(R) & \ge \lambda B(Q')+(1-\lambda)B(Q),
\end{align}
therefore
\begin{equation}
\begin{split}
\entropy(R)+\entropy(R_H)
 & = A(R)+B(R)  \\
 & \ge \lambda(A(P')+B(Q'))+(1-\lambda)(A(P)+B(Q)).
\end{split}
\end{equation}
We use $\entropy(\lambda Q'+(1-\lambda)Q)-\lambda\entropy(Q')-(1-\lambda)\entropy(Q)\le h(\lambda)$ and rearrange to get
\begin{equation}
\begin{split}
2h(\lambda)
 & \ge \entropy(\lambda Q'+(1-\lambda)Q)-\lambda\entropy(Q')-(1-\lambda)\entropy(Q)  \\
 & \qquad +\entropy(\lambda P'_H+(1-\lambda)P_H)-\lambda\entropy(P'_H)-(1-\lambda)\entropy(P_H)  \\
 & \ge \lambda\left[\left(F(G\strongproduct H,P')-F(H,P'_H)\right)-\left(F(G\strongproduct H,Q')-F(H,Q'_H)\right)\right]  \\
 & \qquad +(1-\lambda)\left[\left(F(G\strongproduct H,P)-F(H,P_H)\right)-\left(F(G\strongproduct H,Q)-F(H,Q_H)\right)\right]  \\
 & \ge -\lambda\log f(G)  \\
 & \qquad +(1-\lambda)\left[\left(F(G\strongproduct H,P)-F(H,P_H)\right)-\left(F(G\strongproduct H,Q)-F(H,Q_H)\right)\right],
\end{split}
\end{equation}
which implies \eqref{eq:conditionalFcontinuity}, noting that the roles of $P$ and $Q$ are symmetric.
\end{proof}
\begin{corollary}\label{cor:Fsamemarginalcontinuity}
Let $G,H$ be graphs and $P,Q\in\distributions(V(G)\times V(H))$ such that $P_H=Q_H$. Then
\begin{equation}\label{eq:Fsamemarginalcontinuity}
\left\lvert F(G\strongproduct H,P)-F(G\strongproduct H,Q)\right\rvert\le\delta\log f(G)+2\eta(\delta),
\end{equation}
where $\delta=\frac{1}{2}\norm[1]{P-Q}$.
\end{corollary}
It is crucial that the bounds \eqref{eq:conditionalFcontinuity} and \eqref{eq:Fsamemarginalcontinuity} do not depend on $H$. In contrast, a direct application of the continuity estimate from \cite{vrana2021probabilistic} would only give a bound of the form $\delta\log(\lvert V(G)\rvert\lvert V(H)\rvert-1)+\eta(\delta)$ (or, by a slight modification, $\delta\log f(G\strongproduct H)+\eta(\delta)$).

\begin{proposition}
Let $G$ be a graph.
\begin{enumerate}
\item\label{it:ratecontinuityfinite} If $P,Q\in\distributions(V(G)^n)$, then
\begin{equation}\label{eq:FOrnsteincontinuity}
\left\lvert \frac{1}{n}F(G^{\strongproduct n},P)-\frac{1}{n}F(G^{\strongproduct n},Q)\right\rvert\le \Ornsteindistance(P,Q)\log f(G)+2\eta(\Ornsteindistance(P,Q)).
\end{equation}
\item\label{it:ratecontinuityinfinite} If $\mu$ and $\nu$ are stationary probability measures on $V(G)^\infty$, then
\begin{equation}
\lvert\Frate(G,\mu)-\Frate(G,\nu)\rvert\le \Ornsteindistance(\mu,\nu)\log f(G)+2\eta(\Ornsteindistance(\mu,\nu)).
\end{equation}
\end{enumerate}
\end{proposition}
\begin{proof}
\ref{it:ratecontinuityfinite}:
Let $X_1,\ldots,X_n,Y_1,\ldots,Y_n$ be $V(G)$-valued random variables such that $P$ is the joint distribution of $X_1,\ldots,X_n$, $Q$ is the joint distribution of $Y_1,\ldots,Y_n$, and $\sum_{i=1}^n\probability(X_i\neq Y_i)=n\Ornsteindistance(P,Q)$. Let $R^{(i)}$ denote the joint distribution of $Y_1,Y_2,\ldots,Y_i,X_{i+1},X_{i+2},\ldots X_n$ (in particular, $R^{(0)}=P$ and $R^{(n)}=Q$). Then $\frac{1}{2}\norm[1]{R^{(i)}-R^{(i-1)}}\le\probability(X_i\neq Y_i)$ and the joint distributions of the entries with indices $1,2,\ldots,i-1,i+1,\ldots n$ under $R^{(i)}$ and $R^{(i-1)}$ agree, therefore
\begin{equation}
\begin{split}
\left\lvert F(G^{\strongproduct n},R^{(i)})-F(G^{\strongproduct n},R^{(i-1)})\right\rvert
 & \le \frac{1}{2}{\textstyle\norm[1]{R^{(i)}-R^{(i-1)}}}\log f(G)+2\eta\left({\textstyle\frac{1}{2}\norm[1]{R^{(i)}-R^{(i-1)}}}\right)  \\
 & \le \probability(X_i\neq Y_i)\log f(G)+2\eta(\probability(X_i\neq Y_i))
\end{split}
\end{equation}
by \cref{cor:Fsamemarginalcontinuity}. This implies
\begin{equation}
\begin{split}
\left\lvert \frac{1}{n}F(G^{\strongproduct n},P)-\frac{1}{n}F(G^{\strongproduct n},Q)\right\rvert
 & \le \frac{1}{n}\sum_{i=1}^n \left\lvert F(G^{\strongproduct n},R^{(i)})-F(G^{\strongproduct n},R^{(i-1)})\right\rvert  \\
 & \le \frac{1}{n}\sum_{i=1}^n\big[\probability(X_i\neq Y_i)\log f(G)+2\eta(\probability(X_i\neq Y_i))\big]  \\
 & \le \Ornsteindistance(P,Q)\log f(G)+2\eta(\Ornsteindistance(P,Q)),
\end{split}
\end{equation}
using the Jensen inequality for $\eta$.

\ref{it:ratecontinuityinfinite}:
We apply \eqref{eq:FOrnsteincontinuity} to the inital segments $\mu_{12\ldots n}$ and $\nu_{12\ldots n}$:
\begin{equation}
\left\lvert \frac{1}{n}F(G^{\strongproduct n},\mu_{12\ldots n})-\frac{1}{n}F(G^{\strongproduct n},\nu_{12\ldots n})\right\rvert\le \Ornsteindistance(\mu_{12\ldots n},\nu_{12\ldots n})\log f(G)+2\eta(\Ornsteindistance(\mu_{12\ldots n},\nu_{12\ldots n})).
\end{equation}
The claim follows by taking the limit $n\to\infty$.
\end{proof}

\section{Markov chains}\label{sec:Markov}
The simplest stochastic processes with dependence are Markov chains, characterized by the property that the distribution of each symbol, conditioned on the immediately preceding one, is conditionally independent of all the earlier symbols. The joint distribution of a stationary Markov chain is determined by the distribution of the initial state and the transition matrix. We will assume that the Markov chain is irreducible, in which case the stationary distribution is unique, therefore only the transition probabilities $W(b|a)$ ($a,b\in\mathcal{X}$) need to be specified. The same information is contained in the joint distribution of the first two symbols, which is often a more convenient parametrization: $P(a,b)=P_{12}(a,b)=P_1(a)W(b|a)$, where the first marginal $P_1$ is the stationary distribution (also equal to $P_2$). The corresponding stationary distribution on $\mathcal{X}^\infty$ will be denoted by $\mu^P$.

The probability of a particular string with respect to a time-invariant Markov chain depends only on the first symbol and the numbers of transitions between each pair of symbol. We say that the second-order type of a string $x\in\mathcal{X}^n$ is $Q\in\distributions(\mathcal{X}\times\mathcal{X})$ if $Q$ is the joint type of the length-$(n-1)$ strings $(x_1,\ldots,x_{n-1})$ and $(x_2,\ldots,x_n)$, i.e. $Q(a,b)=\frac{1}{n-1}\lvert\setbuild{i\in[n-1]}{x_i=a,x_{i+1}=b}\rvert$. The probability of $x$ is
\begin{equation}
\begin{split}
\mu^P_{12\ldots n}(x)
 & = P_1(x_1)\prod_{i=1}^{n-1}W(x_{i+1}|x_i)  \\
 & = P_1(x_1)\prod_{a,b}W(b|a)^{(n-1)Q(a,b)}  \\
 & = P_1(x_1)2^{(n-1)\sum_{a,b}Q(a,b)\log W(b|a)}  \\
 & = P_1(x_1)2^{(n-1)\sum_{a,b}Q(a,b)(\log P(a,b)-\log P_1(a))}  \\
 & = P_1(x_1)2^{-(n-1)\left[\entropy(Q)+\relativeentropy{Q}{P}-\entropy(Q_1)-\relativeentropy{Q_1}{P_1}\right]}  \\
 & = P_1(x_1)2^{-(n-1)\left[\entropy(2|1)_Q+\relativeentropy{Q_{2|1}}{W}\right]}.
\end{split}
\end{equation}
The set of possible second-order types of length-$n$ sequences $x$ over the alphabet $\mathcal{X}$, with first and last letters $x_1=a$ and $x_n=b$ will be denoted by $\secondorderdistributions[n](\mathcal{X},a,b)$, and for $Q\in\secondorderdistributions[n](\mathcal{X},a,b)$ we denote by $\secondordertypeclass{n}{Q,a,b}$ the set of strings $x$ having second-order type $Q$ and starting with $a$ and ending with $b$. The union of the sets $\secondorderdistributions[n](\mathcal{X},a,b)$ over $a,b\in\mathcal{X}$ will be denoted by $\secondorderdistributions[n](\mathcal{X})$.

We denote by $\secondorderdistributions(\mathcal{X})$ the set of probability distributions on $\mathcal{X}$ with equal marginals and such that the Markov chain $\mu^P$ is irreducible, when considered on the state space $\support P_1$.

For a subset $\mathcal{U}\subseteq\distributions(\mathcal{X}\times\mathcal{X})$ we introduce the notation $\secondordertypeclass{n}{\mathcal{U}}=\bigcup_{a,b\in\mathcal{X}}\bigcup_{Q\in\secondorderdistributions[n](\mathcal{X},a,b)\cap\mathcal{U}}\secondordertypeclass{n}{Q,a,b}$, and use similar notation for unions of first-order type classes. In particular, the role of $\mathcal{U}$ will be played by an $\epsilon$-ball $\ball{\epsilon}{P}$ in $\distributions(\mathcal{X}\times\mathcal{X})$ with respect to the total variation distance.

The basic estimates for second-order type classes are summarized in the following lemma (see \cite{boza1971asymptotically}, \cite{whittle1955some,billingsley1961statistical} for exact values, and the review \cite{csiszar1998method} on various type concepts and applications).
\begin{lemma}\label{lem:typeclassestimates}
Let $\mathcal{X}$ be a finite set, $n\in\naturals$, and $Q\in\secondorderdistributions[n](\mathcal{X})$, $a,b\in\support Q_1$.
\begin{enumerate}
\item $\lvert\secondorderdistributions[n](\mathcal{X},a,b)\rvert\le n^{\lvert\mathcal{X}\rvert^2}$;
\item $\lvert\secondordertypeclass{n}{Q,a,b}\rvert=2^{n\entropy(2|1)_Q+o(n)}$;
\item $\mu^P_{12\ldots n}(\secondordertypeclass{n}{Q,a,b})=2^{-n\relativeentropy{Q_{2|1}}{W}+o(n)}$.
\end{enumerate}
\end{lemma}

\begin{proposition}\label{prop:MCbounds}
Let $G$ be a graph and $P\in\secondorderdistributions(V(G))$, $f\in\Delta(\graphs)$ and $c\in(0,1)$. Then
\begin{enumerate}
\item\label{it:MCupperbound} $\displaystyle\limsup_{n\to\infty}\min_{\substack{S\subseteq V(G)^n  \\  \mu^P_{12\ldots n}(S)\ge c}}\frac{1}{n}\log f(G^{\strongproduct n}[S])\le F(G,P_1),
$
\item\label{it:MClowerbound} $\displaystyle\liminf_{n\to\infty}\min_{\substack{S\subseteq V(G)^n  \\  \mu^P_{12\ldots n}(S)\ge c}}\frac{1}{n}\log f(G^{\strongproduct n}[S])\ge F(G,P_1)-\mutualinformation(1:2)_P.
$
\end{enumerate}
\end{proposition}
\begin{proof}
\ref{it:MCupperbound}:
By the law of large numbers, for every $\epsilon>0$, the set $S=\typeclass{n}{\ball{\epsilon}{P_1}}$ has probabilty $1-o(1)$ under $\mu_{12\ldots n}$, as $n\to\infty$. This implies
\begin{equation}
\begin{split}
\limsup_{n\to\infty}\min_{\substack{S\subseteq V(G)^n  \\  \mu^P_{12\ldots n}(S)\ge c}}\frac{1}{n}\log f(G^{\strongproduct n}[S])
 & \le \lim_{\epsilon\to 0}\limsup_{n\to\infty}\frac{1}{n}\log f(\typegraph{G}{n}{\ball{\epsilon}{P_1}})  \\
 & = F(G,P_1)
\end{split}
\end{equation}
by \cite[Proposition 3.4]{vrana2021probabilistic}.

\ref{it:MClowerbound}:
Let $\epsilon>0$, $n\in\naturals$, and let $S\subseteq V(G)^n$ such that $\mu^P_{12\ldots n}(S)\ge c$. Then
\begin{equation}\label{eq:intersectionprobabilitylowerbound}
\mu^P_{12\ldots n}(S\cap\secondordertypeclass{n}{\ball{\epsilon}{P}})\ge\mu^P_{12\ldots n}(S)-\left(1-\mu^P_{12\ldots n}(\secondordertypeclass{n}{\ball{\epsilon}{P}})\right)\ge c-o(1)
\end{equation}
by the law of large numbers. Choose $a,b\in\mathcal{X}$ and $Q\in\ball{\epsilon}{P}$ such that $\mu^P_{12\ldots n}(S\cap\secondordertypeclass{n}{Q,a,b})$ is maximal. Since $\mu^P_{12\ldots n}(S\cap\secondordertypeclass{n}{\ball{\epsilon}{P}})$ is the sum of at most $n^{\lvert V(G)\rvert^2}\lvert V(G)\rvert^2$ such terms, the maximal probability is at least $n^{-\lvert V(G)\rvert^2}\lvert V(G)\rvert^{-2}\mu^P_{12\ldots n}(S\cap\secondordertypeclass{n}{\ball{\epsilon}{P}})$. This implies that the part of this second-order type class covered by $S$ is not too small:
\begin{equation}\label{eq:intersectionnottoosmall}
\frac{\lvert S\cap\secondordertypeclass{n}{Q,a,b}\rvert}{\lvert \secondordertypeclass{n}{Q,a,b}\rvert}=\frac{\mu^P_{12\ldots n}(S\cap\secondordertypeclass{n}{Q,a,b})}{\mu^P_{12\ldots n}(\secondordertypeclass{n}{Q,a,b})}\ge\frac{\mu^P_{12\ldots n}(S\cap\secondordertypeclass{n}{\ball{\epsilon}{P}})}{n^{\lvert V(G)\rvert^2}\lvert V(G)\rvert^2},
\end{equation}
where we used that the restriction of $\mu^P_{12\ldots n}$ to $\secondordertypeclass{n}{Q,a,b}$ is uniform.

Note that the first-order type of a string is a function of its second-order type and the first and last letter. More precisely, let $\tilde{Q}=\frac{n-1}{n}Q_1+\frac{1}{n}\delta_b$ where $\delta_b$ is the Dirac measure concentrated at $b$. Then $\secondordertypeclass{n}{Q,a,b}\subseteq\typeclass{n}{\tilde{Q}}$.
Since $\typegraph{G}{n}{\tilde{Q}}$ is vertex-transitive, we can use \cite[Lemma 3.1]{vrana2021probabilistic} to conclude
\begin{equation}
\typegraph{G}{n}{\tilde{Q}}\le\complement{K_N}\strongproduct G^{\strongproduct n}[S\cap\secondordertypeclass{n}{Q,a,b}]\le \complement{K_N}\strongproduct G^{\strongproduct n}[S],
\end{equation}
where
\begin{equation}
\begin{split}
N
 & = \left\lfloor\frac{\lvert\typeclass{n}{\tilde{Q}}\rvert\ln\lvert\typeclass{n}{\tilde{Q}}\rvert}{\lvert S\cap\secondordertypeclass{n}{Q,a,b}\rvert}\right\rfloor+1  \\
 & \le \frac{\lvert\typeclass{n}{\tilde{Q}}\rvert\ln\lvert\typeclass{n}{\tilde{Q}}\rvert}{\lvert \secondordertypeclass{n}{Q,a,b}\rvert}\frac{n^{\lvert V(G)\rvert^2}\lvert V(G)\rvert^2}{\mu^P_{12\ldots n}(S\cap\secondordertypeclass{n}{\ball{\epsilon}{P}})}+1  \\
 & \le \frac{2^{n\entropy(\tilde{Q})}n\ln\lvert V(G)\rvert}{2^{n\entropy(2|1)_Q+o(n)}}\frac{n^{\lvert V(G)\rvert^2}\lvert V(G)\rvert^2}{c-o(1)}+1  \\
 & = 2^{n\mutualinformation(1:2)_Q+o(n)}.
\end{split}
\end{equation}
Here the first inequality uses \eqref{eq:intersectionnottoosmall}, and the second inequality uses \eqref{eq:intersectionprobabilitylowerbound}, \cref{lem:typeclassestimates} and that $\lvert\typeclass{n}{\tilde{Q}}\rvert\le 2^{n\entropy(\tilde{Q})}\le\lvert V(G)\rvert^n$, while the last equality uses that $\tilde{Q}\to Q_1=Q_2$ as $n\to\infty$. This implies
\begin{equation}
\frac{1}{n}\log f(\typegraph{G}{n}{\tilde{Q}})\le\mutualinformation(1:2)_Q+o(1)+\frac{1}{n}\log f(G^{\strongproduct n}[S]),
\end{equation}
therefore
\begin{equation}
\liminf_{n\to\infty}\min_{\substack{S\subseteq V(G)^n  \\  \mu^P_{12\ldots n}(S)\ge c}}\frac{1}{n}\log f(G^{\strongproduct n}[S])\ge \min_{Q\in\closedball{\epsilon}{P}}\left[F(G,Q_1)-\mutualinformation(1:2)_Q\right],
\end{equation}
using again that $\tilde{Q}\to Q_1$ as $n\to\infty$. The inequality holds for every $\epsilon>0$ and both $F(G,\cdot)$ and the mutual information are continuous, therefore the stated inequality follows by taking the limit as $\epsilon\to 0$.
\end{proof}

\begin{theorem}\label{thm:MarkovAEP}
Let $G$ be a graph and $P\in\secondorderdistributions(V(G))$. Then the pair $(G,\mu^P)$ has the asymptotic equipartition property. More precisely, for every $f\in\Delta(\graphs)$, with probabilistic refinement $F$, the equality
\begin{equation}
\lim_{n\to\infty}\min_{\substack{S\subseteq V(G)^n  \\  \mu^P_{12\ldots n}(S)\ge c}}\frac{1}{n}\log f(G^{\strongproduct n}[S])=\Frate(G,\mu^P)
\end{equation}
holds for all $c\in(0,1)$.
\end{theorem}
\begin{proof}
Let $k\in\positiveintegers$ be coprime to the period of $\mu^P$ and consider blocks of $k$ outputs from the source. The resulting distribution is of the form $\mu^{P^{(k)}}$ for some $P^{(k)}\in\secondorderdistributions(V(G)^k)$, and satisfies $\mutualinformation(1:2)_{P^{(k)}}=\mutualinformation(1:2)_P$ by the Markov property, and $\Frate(G^{\strongproduct k},P^{(k)})=k\Frate(G,P)$. Using \cref{prop:blockAEP,prop:MCbounds} we get
\begin{equation}
\begin{split}
\limsup_{n\to\infty}\min_{\substack{S\subseteq V(G)^n  \\  \mu^P_{12\ldots n}(S)\ge c}}\frac{1}{n}\log f(G^{\strongproduct n}[S])
 & = \limsup_{n\to\infty}\min_{\substack{S\subseteq V(G)^{kn}  \\  \mu^{P^{(k)}}_{12\ldots n}(S)\ge c}}\frac{1}{kn}\log f(G^{\strongproduct kn}[S])  \\
 & \le \frac{1}{k}F(G^{\strongproduct k},P^{(k)}_1)  \\
 & = \frac{1}{k}F(G^{\strongproduct k},\mu^P_{12\ldots k})
\end{split}
\end{equation}
and
\begin{equation}
\begin{split}
\liminf_{n\to\infty}\min_{\substack{S\subseteq V(G)^n  \\  \mu^P_{12\ldots n}(S)\ge c}}\frac{1}{n}\log f(G^{\strongproduct n}[S])
 & = \liminf_{n\to\infty}\min_{\substack{S\subseteq V(G)^{kn}  \\  \mu^{P^{(k)}}_{12\ldots n}(S)\ge c}}\frac{1}{kn}\log f(G^{\strongproduct kn}[S])  \\
 & \ge \frac{1}{k}F(G^{\strongproduct k},P^{(k)}_1)-\frac{1}{k}\mutualinformation(1:2)_{P^{(k)}}  \\
 & = \frac{1}{k}F(G^{\strongproduct k},\mu^P_{12\ldots k})-\frac{1}{k}\mutualinformation(1:2)_{P}.
\end{split}
\end{equation}
As $k\to\infty$, both bounds converge to $\Frate(G,\mu^P)$.
\end{proof}
The result is easily extended to higher-order Markov chains by \cref{prop:blockAEP}, using that the induced process on sufficiently large blocks is a first-order Markov chain. Alternatively, one may arrive at the same conclusion from \cref{thm:hiddenprocess} below, by converting a higher-order Markov chain to a first-order hidden Markov model.

In the memoryless case the value $\Frate(G,P^{\otimes\infty})=F(G,P)$ can also be characterized as the limit $\lim_{\epsilon\to 0}\lim_{n\to\infty}\frac{1}{n}\log f(G^{\strongproduct n}[\typeclass{n}{\ball{\epsilon}{P}}])$. We do not know if the analogous expression for Markov chains is equal to the $F$-rate.
\begin{problem}
Is $\Frate(G,\mu^P)$ equal to the limit $\displaystyle\lim_{\epsilon\to 0}\lim_{n\to\infty}\frac{1}{n}\log f(G^{\strongproduct n}[\secondordertypeclass{n}{\ball{\epsilon}{P}}])$?
\end{problem}

We conclude this section with a generalization K\"orner's coding theorem \cite{korner1973coding} to irreducible Markov chains with an ambiguous alphabet, identifying the graph entropy rate as the number of bits per letter required for an encoding that maps distinguishable symbols to distinct codewords. For clarity and to ease comparison with \cite{korner1973coding}, where the i.i.d. case is considered, we use the complementary (relative to the rest of this paper) picture, where edges correspond to \emph{distinguishable} pairs of symbols. Recall that the chromatic number of a graph has the coding interpretation as the minimal output cardinality for a valid encoding.
\begin{proposition}
Let $G$ be a graph and $P\in\secondorderdistributions(V(G))$. Then for every $c\in(0,1)$ the equality
\begin{equation}
\lim_{n\to\infty}\min_{\substack{S\subseteq V(G)^n  \\  \mu^P_{12\ldots n}(S)\ge c}}\frac{1}{n}\log \chi(G^{\costrongproduct n}[S])=\lim_{n\to\infty}\frac{1}{n}\entropy(G,\mu^P_{12\ldots n}).
\end{equation}
holds.
\end{proposition}
\begin{proof}
The graph entropy of the complement is the probabilistic refinement of the fractional clique cover number, therefore by \cref{thm:MarkovAEP} we have
\begin{equation}
\begin{split}
\lim_{n\to\infty}\min_{\substack{S\subseteq V(G)^n  \\  \mu^P_{12\ldots n}(S)\ge c}}\frac{1}{n}\log \chi_f(G^{\costrongproduct n}[S])
 & = \lim_{n\to\infty}\min_{\substack{S\subseteq V(G)^n  \\  \mu^P_{12\ldots n}(S)\ge c}}\frac{1}{n}\log \complement{\chi}_f(\complement{G}^{\strongproduct n}[S])  \\
 & =\lim_{n\to\infty}\frac{1}{n}\entropy(G,\mu^P_{12\ldots n}).
\end{split}
\end{equation}

The chromatic number and the fractional chromatic number are related as \cite{lovasz1975ratio}
\begin{equation}
\chi_f(H)\le\chi(H)\le\chi_f(H)(1+\ln\alpha(H))\le\chi_f(H)(1+\ln\lvert V(H)\rvert).
\end{equation}
This implies that
\begin{equation}
\frac{1}{n}\log \chi_f(G^{\costrongproduct n}[S])\le\frac{1}{n}\log \chi(G^{\costrongproduct n}[S])\le\frac{1}{n}\log \chi_f(G^{\costrongproduct n}[S])+\frac{\log(1+n\ln\lvert V(G)\rvert)}{n},
\end{equation}
therefore
\begin{equation}
\lim_{n\to\infty}\min_{\substack{S\subseteq V(G)^n  \\  \mu^P_{12\ldots n}(S)\ge c}}\frac{1}{n}\log \chi(G^{\costrongproduct n}[S])
 = \lim_{n\to\infty}\min_{\substack{S\subseteq V(G)^n  \\  \mu^P_{12\ldots n}(S)\ge c}}\frac{1}{n}\log \chi_f(G^{\costrongproduct n}[S]).
\end{equation}
\end{proof}

\section{Functions of processes}\label{sec:hidden}

In this section we elevate \cref{ex:HMM} to what appears to be the ``right'' generality within the present framework. We generalize in two ways: we will not put any restriction on the process $\mu$, and allow the observed output to be partially distinguishable. The alphabet for the ``hidden'' states will be $\mathcal{Z}$, $\mu$ is a probability measure on $\mathcal{Z}^\infty$, $G$ is the confusability graph of the output alphabet, and $\varphi:\mathcal{Z}\to V(G)$ is the function that determines the output from the hidden state. The main observation is that this model is equivalent (in the precise sense formulated below) to allowing the hidden process to be directly observed, but equipped with a confusability relation determined by $G$ and $\varphi$. The details are as follows.
\begin{definition}
Let $G$ be a graph and $\varphi:\mathcal{Z}\to V(G)$ a function from an arbitrary set $\mathcal{Z}$. We define the pullback $\varphi^*(G)$ as the graph with vertex set $\mathcal{Z}$ and such that there is an edge between distinct vertices $z_1,z_2\in \mathcal{Z}$ iff $\varphi(z_1)$ and $\varphi(z_2)$ are adjacent or $\varphi(z_1)=\varphi(z_2)$. Equivalently: $z_1\adjacentorequal z_2\iff\varphi(z_1)\adjacentorequal\varphi(z_2)$.
\end{definition}
Clearly, $\varphi^*(G)$ has the smallest edge set (with respect to containment) on $\mathcal{Z}$ that makes $\varphi$ a homomorphism between the complements. We can use this construction to relate the asymptotic equipartition property for the observed process with a given confusability graph to that of the hidden process.
\begin{lemma}\label{lem:pullbackF}
Let $G$ be a graph, $\varphi:\mathcal{Z}\to V(G)$ a function, and $P\in\distributions(\mathcal{Z})$. If $F$ is the probabilistic refinement of a spectral point, then
\begin{equation}
F(G,\varphi_*(P))=F(\varphi^*(G),P).
\end{equation}
\end{lemma}
\begin{proof}
Since $\varphi:\complement{\varphi^*(G)}\to\complement{G}$ is a homomorphism, we have $F(\varphi^*(G),P)\le F(G,\varphi_*(P))$ by \cref{it:Fmonotone}.

Conversely, let $\psi$ be any right inverse of $\varphi$ on $\varphi(\mathcal{Z})$. Then $\psi:\complement{G[\varphi(\mathcal{Z})]}\to\complement{\varphi^*(G)}$ is a homomorphism, therefore
\begin{equation}
F(G,\varphi_*(P))=F(G[\varphi(\mathcal{Z})],\varphi_*(P))\le F(\varphi^*(G),(\varphi\circ\psi)_*(P)).
\end{equation}
The claim follows by concavity of $F$, using that $P$ is in the convex hull of the probability distributions of the form $(\varphi\circ\psi)_*(P)$.
\end{proof}
\begin{remark}
The reader may recognize the pullback graph as a special case of the $G$-join \cite[(6.1) Definition]{sabidussi1961graph} of a family of complete graphs, one for each fiber of $\varphi$. A generalization of \cref{lem:pullbackF} to arbitrary $G$-joins is proved in \cite[Proposition 4.4]{vrana2021probabilistic} under the additional condition that $f$ is multiplicative under the lexicographic product. While every known element of $\Delta(\graphs)$ has this property, it is an open question whether it is true for every spectral point \cite[Problem 6.9]{fritz2020unified}. Precisely this special case is used in \cite[Theorem 4.8]{vrana2021probabilistic} for proving that the complementary function $\log f^*(G)=\max_{P\in\distributions(V(G))}\left[\entropy(P)-F(\complement{G},P)\right]$ of such spectral points is $\le$-monotone, therefore \cref{lem:pullbackF} proves this monotonicity unconditionally, for the complementary function of every spectral point.
\end{remark}

\begin{corollary}\label{cor:pullbackinducedsubgraphf}
Let $T\subseteq\mathcal{Z}$. Then $f(\varphi^*(G)[T])=f(G[\varphi(T)])$.
\end{corollary}
\begin{proof}
$\log f(\varphi^*(G)[T])$ is equal to the maximum of $F(\varphi^*(G),P)$ over those distributions $P$ which are supported within $T$. The set of the pushforwards of these distributions along $\varphi$ are precisely the distributions supported within $\varphi(T)$. By \cref{lem:pullbackF}, the maximum is the same as the maximum of $F(G,\varphi_*(P))$, which is equal to $\log f(G[\varphi(T)])$.
\end{proof}

\begin{lemma}\label{lem:pullbacktypicalf}
Let $G$ be a graph, $\varphi:\mathcal{Z}\to V(G)$ a function, $P\in\distributions(\mathcal{Z})$, and $c\in(0,1)$. Then
\begin{equation}
\min_{\substack{S\subseteq V(G)  \\  \varphi_*(P)(S)\ge c}}f(G[S])=\min_{\substack{T\subseteq \mathcal{Z}  \\  P(T)\ge c}}f(\varphi^*(G)[T]).
\end{equation}
\end{lemma}
\begin{proof}
Let $S\subseteq V(G)$ such that $\varphi_*(P)(S)\ge c$. Then $P(\varphi^{-1}(S))=\varphi_*(P)(S)\ge c$, therefore $T=\varphi^{-1}(S)$ is a feasible point for the minimum on the right hand side. This implies
\begin{equation}
\begin{split}
\smash{\min_{\substack{T\subseteq \mathcal{Z}  \\  P(T)\ge c}}}f(\varphi^*(G)[T])
 & \le f(\varphi^*(G)[\varphi^{-1}(S)])  \\
 & = f(G[\varphi(\varphi^{-1}(S))])  \\
 & = f(G[S]),
\end{split}
\end{equation}
where the first equality uses \cref{cor:pullbackinducedsubgraphf}. We minimize over $S$ to get the inequality $\ge$.

Conversely, let $T\subseteq\mathcal{Z}$ such that $P(T)\ge c$. Since $\varphi^{-1}(\varphi(T))\supseteq T$, we also have $\varphi_*(P)(\varphi(T))\ge c$, i.e. $S=\varphi(T)$ is a feasible point for the left hand side, therefore
\begin{equation}
\begin{split}
\smash{\min_{\substack{S\subseteq V(G)  \\  \varphi_*(P)(S)\ge c}}}f(G[S])
 & \le f(G[\varphi(T)])  \\
 & = f(\varphi^*(G)[T]),
\end{split}
\end{equation}
again using \cref{cor:pullbackinducedsubgraphf}. Taking the minimum over $T$ gives the inequality $\le$.
\end{proof}

\begin{lemma}\label{lem:pullbackstrongproduct}
Let $G_1$, $G_2$ be graphs and $\varphi_1:\mathcal{Z}_1\to V(G_1)$, $\varphi_2:\mathcal{Z}_2\to V(G_2)$ functions. Then $\varphi_1^*(G_1)\strongproduct\varphi_2^*(G_2)=(\varphi_1\times\varphi_2)^*(G_1\strongproduct G_2)$.
\end{lemma}
\begin{proof}
Let $z_1,w_1\in \mathcal{Z}_1$ and $z_2,w_2\in\mathcal{Z}_2$. Then
\begin{multline}
(z_1,z_2)\adjacentorequal(w_1,w_2)\text{ in }\varphi_1^*(G_1)\strongproduct\varphi_2^*(G_2)  \\
\iff \big(z_1\adjacentorequal w_1\text{ in }\varphi_1^*(G_1)\big)\text{ and }\big(z_2\adjacentorequal w_2\text{ in }\varphi_2^*(G_2)\big)  \\
\iff \big(\varphi_1(z_1)\adjacentorequal\varphi_1(w_1)\text{ in }G_1\big)\text{ and }\big(\varphi_2(z_2)\adjacentorequal\varphi_2(w_2)\text{ in }G_2\big)  \\
\iff (\varphi_1(z_1),\varphi_2(z_2))\adjacentorequal(\varphi_1(w_1),\varphi_2(w_2))\text{ in }G_1\strongproduct G_2  \\
\iff (z_1,z_2)\adjacentorequal(w_1,w_2)\text{ in }(\varphi_1\times\varphi_2)^*(G_1\strongproduct G_2).
\end{multline}
\end{proof}

\begin{theorem}\label{thm:hiddenprocess}
Let $\mathcal{Z}$ be a (finite) alphabet, $G$ a graph, $\varphi:\mathcal{Z}\to V(G)$ a function, and $\mu$ a probability measure on $\mathcal{Z}^\infty$. Then
\begin{enumerate}
\item\label{it:hiddenAEP} $(G,\varphi^\infty_*(\mu))$ has the asymptotic equipartition property iff $(\varphi^*(G),\mu)$ does;
\item\label{it:hiddenFrate} if either $\Frate(G,\varphi^\infty_*(\mu))$ or $\Frate(\varphi^*(G),\mu)$ exist, then the other one exists as well and the two are equal.
\end{enumerate}
\end{theorem}
\begin{proof}
\ref{it:hiddenAEP}:
For every $c\in(0,1)$ and $n\in\naturals$ we have
\begin{equation}
\min_{\substack{S\subseteq V(G)^n  \\  \varphi^\infty_*(\mu)_{12\ldots n}(S)\ge c}}\frac{1}{n}\log f(G^{\strongproduct n}[S])=\min_{\substack{T\subseteq \mathcal{Z}^n  \\  \mu_{12\ldots n}(T)\ge c}}\frac{1}{n}\log f(\varphi^*(G)^{\strongproduct n}[T])
\end{equation}
by \cref{lem:pullbacktypicalf,lem:pullbackstrongproduct}. The claim follows by taking the limit $n\to\infty$.

\ref{it:hiddenFrate}:
By \cref{lem:pullbackF,lem:pullbackstrongproduct} we have
\begin{equation}
\begin{split}
F(G^{\strongproduct n},\varphi^\infty_*(\mu)_{12\ldots n})
 & = F(G^{\strongproduct n},\varphi^n_*(\mu_{12\ldots n}))  \\
 & = F(\varphi^n_*(G^{\strongproduct n}),\mu_{12\ldots n})  \\
 & = F(\varphi_*(G)^{\strongproduct n},\mu_{12\ldots n})
\end{split}
\end{equation}
for all $n\in\naturals$. After dividing by $n$ and taking the limit (if it exist), we obtain $\Frate(G,\varphi^\infty_*(\mu))=\Frate(\varphi^*(G),\mu)$.
\end{proof}

In particular, if $\mu$ is an irreducible stationary Markov chain, then $(G,\mu)$ has the asymptotic equipartition property, and the limit in \cref{def:generalAEP} is the $F$-rate by \cref{thm:MarkovAEP}. From \cref{thm:hiddenprocess} we can conclude that the same is true when $\mu$ is a hidden Markov model such that the hidden process is irreducible stationary.

\section*{Acknowledgement}

This work was partially funded by the National Research, Development and Innovation Office of Hungary via the research grants K~124152, by the \'UNKP-21-5 New National Excellence Program of the Ministry for Innovation and Technology, the J\'anos Bolyai Research Scholarship of the Hungarian Academy of Sciences, and by the Ministry of Innovation and Technology and the National Research, Development and Innovation Office within the Quantum Information National Laboratory of Hungary.

\bibliography{refs}{}

\end{document}